\documentclass[peerreview,12pt,onecolumn,draftclsnofoot]{IEEEtran}

\usepackage{amssymb,amsmath}
\usepackage{amsthm}
\usepackage{mathrsfs}
\usepackage{graphicx}
\usepackage{epic, eepic}
\usepackage{multirow}
\usepackage{cite}
\usepackage{enumerate}
\usepackage{threeparttable}
\usepackage{comment}

\newtheorem{lem}{Lemma}

\newcommand{\E}[1]{\textsf{E}\left[{#1}\right]}   
\newcommand{\Ep}[2]{\textsf{E}_{#1}\left[{#2}\right]}   
\newcommand{\set}[1]{\mathcal{#1}}
\newcommand{\eqdef}{\; \hat{=} \;}

\newtheorem{theorem}{Theorem}
\newtheorem*{theorem*}{Theorem}
\begin{document}
\title{\Large\bfseries Optimal Training for Non-Feedback Adaptive
  Modulation over Rayleigh Fading Channels}

\author{%
  Khalid Zeineddine$^1$, Hussein Hammoud$^2$, Ibrahim
  Abou-Faycal$^2$\thanks{This work was performed at AUB and supported
    by AUB's University Research Board.}  \thanks{Partial results of
    this study were presented at the IEEE ISSPIT'09, Dec. 2009.}
  \\
  $^1$ {\small Department of Electrical Engineering \& Computer
    Science, Northwestern University, USA}
  \\
  $^2$ {\small Department of Electrical \& Computer Engineering,
    American University of Beirut, Lebanon}
  \\
  {\small {\tt khalidzeineddine2015@u.northwestern.edu,
      \{hmh52, ia14\}@aub.edu.lb}} }

\maketitle

\begin{abstract}

  
  Time-varying fading channels present a major challenge in the design
  of wireless communication systems. Adaptive schemes are often
  employed to adapt the transmission parameters to receiver-based
  estimates of the quality of the channel.
  We consider a pilot-based adaptive modulation scheme without the use
  of a feedback link. In this scheme, pilot tones (known by sender and
  receiver) are periodically sent through the channel for the purpose
  of channel estimation and coherent demodulation of data symbols at
  the receiver. We optimize the duration and power allocation of these
  pilot symbols to maximize the information-theoretic achievable rates
  using binary signaling. We analyze four transmission policies and
  numerically show how optimal training in terms of duration and power
  allocation varies with the channel conditions and from one
  transmission policy to another. We prove that for a causal
  estimation scheme with flexible power allocation, placing all the
  available power on one pilot is optimal.
 
\end{abstract}

\begin{keywords}
  Adaptive modulation, pilot symbol assisted modulation, fading
  channels, Rayleigh fading, power allocation, training duration.
\end{keywords}


\section{Introduction}

In digital mobile communications, fast fading degrades the Bit Error
Rate (BER) of the channel and inhibits coherent
detection\footnote{Coherent demodulation requires the extraction of a
  reliable phase reference from the received signal.}.  It is known
that performance is limited by channel estimation errors~\cite
{chester1, chester2, chester3, chester4}. Pilot Symbol Assisted
Modulation (PSAM) is a technique that has been introduced
in~\cite{Cav91} to mitigate these effects. In this scheme, known
training symbols (pilots) are periodically inserted into the data
frame for the purpose of channel estimation and coherent demodulation
of the data symbols.

Furthermore, channel-adaptive modulation dynamically adjusts certain
transmission parameters such as the constellation size, transmitted
power, and code rate according to the channel quality. Adaptive
signaling provides in general higher bit rates (relative to
conventional nonadaptive methods) by increasing the transmission
throughput under favorable channel conditions and reducing it as the
channel condition is degraded.



Some of the previous adaptive schemes rely on a channel-feedback link
to provide the transmitter with the Channel Side Information
(CSI)~\cite{Gold97, Cai05}. In~\cite{ref1}, the authors consider
employment of adaptive modulation with one pilot in addition to
delayed feedback to the transmitter and prove that power adaptation
via periodic feedback can increase the achievable rates. Similarly,
in~\cite{ref3}, authors consider pilot-based adaptive modulation where
estimate is fed back to transmitter in order to adapt data and pilot
power and study the optimal policy for power allocation for data and
pilot symbols. Authors in \cite{ref18} discuss adaptive modulation
with feedback and develop an adaptive scheme that accounts for both
channel estimation and prediction errors in order to meet a target Bit
Error Rate (BER). In \cite{ref16}, the authors attempt to optimize the
spectral efficiency subject to a specific BER constraint in a
pilot-based adaptive modulation setup with feedback. The above
mentioned works study the performance of such systems and prove
adaptive modulation using pilots can increase the achievable rates in
general. However, systems that rely on a channel-feedback link present
some disadvantages because of the modeling complexity on one hand and
its infeasibility on the other hand when the channel is fading faster
than it can be estimated (or predicted) and fed back to the sender.
Optimizing the pilot placement, power allocation and modulation
schemes in a pilot-based setup is an active area of research, whether
in the case of a single receiver~\cite{ref1, ref3, ref4, ref6, ref9,
  ref10, ref11, ref13, ref18} or multiple receivers~\cite{ref13,
  ref15, ref19}.

A modified pilot-based adaptive modulation scheme over Rayleigh fading
channels was presented in~\cite{AMM05}. This scheme adapts the coded
modulation strategy at the sender to the quality of the channel
estimation (estimation error variance) at the receiver \textit{without
  requiring any channel feedback}. In this work we study the
performance of this non-feedback adaptive modulation scheme over
time-varying Rayleigh fading channels. Unlike the scheme
in~\cite{AMM05}, we consecutively send a cluster of $k$ pilots ($k
\geq 1$) per data frame with $k$ being an optimization
variable~\cite{ZAF09}. We determine the optimal duration and power
allocation of the training period under different transmission
policies for both causal and non-causal estimation.
We study such systems at low Signal-to-Noise-Ratio (SNR) (we consider
the received SNR) levels and the performance is measured in terms of
achievable rates using binary signaling. We prove that the ``optimal''
power allocation scheme which minimizes the error variance of the
estimates of the channel parameters --\textit{which is set up offline
  without requiring feedback}-- in case of causal estimation is the
one in which all the available power is allocated on one pilot, if
constraints allow it.


The organization of this paper is as follows. In Section~\ref{sc:pre},
we present the fading channel model, the adaptive transmission
technique we use to transmit over the channel as well as the receiver
details. The measure of performance is discussed in
Section~\ref{sc:MI}, the optimal power allocation for causal
estimation is proved in Section~\ref{sc:OPT} and the numerical results
are presented in Section~\ref{sc:NR}. In Section~\ref{sc:Ext}, we
present possible extensions to other fading models and
Section~\ref{sc:conc} concludes the paper.

\section{Preliminaries}
\label{sc:pre}

\subsection{The Channel Model}

Consider the single-user discrete-time model for the Rayleigh fading
channel,
\begin{equation*}
  Y_{i} = R_{i}X_{i}+N_{i},
\end{equation*}
where $i$ is the time index, $X_{i} \in \mathbb{C}$ is the channel
input at time $i$, $Y_{i} \in \mathbb{C}$ is its output, and $R_{i}$
and $N_{i}$ are independent complex circular Gaussians\footnote{A
  complex Gaussian random variable is circular if and only if it is
  zero-mean and its real and imaginary parts are independent with
  equal variances.}  random variables with zero mean and variance
$\sigma _{R}^{2}$ and $\sigma _{N}^{2}$ respectively.  The amplitude
of the fading coefficient $R_{i}$ is then Rayleigh distributed and its
phase is uniform over $[-\pi,\,\pi)$. To account for power
constraints, the input is subject to
\begin{equation*}
  \E{ \left\vert X_{i}\right\vert ^{2} } \leq P_i,
\end{equation*}
for some parameters $\{P_i\}$ --that could be all equal to a constant
for example.  Since from an information theoretic perspective scaling
the output by $1 / \sigma_R$ does not change the mutual information,
we assume without loss of generality that $\sigma_R = 1$. The variance
of the noise $\sigma_{N}^2$ is to be generally interpreted as
$(\sigma_{N}^2 / \sigma_{R}^2)$.


We assume in this study that the fading process follows a stationary
first-order Gauss-Markov model introduced in~\cite{Med00}, i.e.,
\begin{equation}
  R_{i} = \alpha R_{i-1}+Z_{i},
  \label{eqn:model}
\end{equation}
where the samples $\{ Z_{i}\}$ are Independent and
Identically-Distributed (IID) complex circular Gaussians with mean
zero and variance equal to $\sigma_{Z}^{2} = \left( 1-\alpha^{2}
\right)$ such that $\alpha \in \left[0,1\right)$ to guarantee
stationarity.

Even though we analyze the benefits of pilot clustering by assuming
that the autocorrelation function of the fading process is derived
from a stationary first-order Gauss-Markov model~(\ref{eqn:model}), we
argue in Section~\ref{sc:Ext} that the methodology may be readily
adapted to other models and present the case of a Jakes' model
\cite{JakesBook} that takes into account higher orders of correlation.

\subsection{The Adaptive Transmission Scheme}

At regular intervals, the transmitter successively sends $k$ known
pilot symbols whose purpose is to enable the estimation of the channel
at the receiver.  The channel estimation is solely based on the pilot
symbols and no data-directed estimation is used. For each time sample
$i$, the receiver computes the Minimum Mean-Square Estimate (MMSE) of
the channel, the quality of which --measured through the estimate
error variance-- depends on its position with respect to the pilot
symbols. After estimation, the channel, as seen by the receiver, is a
Rician channel whose specular part is given by the estimate and whose
Rayleigh component is given by the zero-mean Gaussian-distributed
estimation error.

Although the scheme is adaptive, it {\em does not use feedback\/} to
determine its policy. The key idea is that the transmitter adapts to
the quality of channel estimation (specifically to the mean-square
error which is independent of the value of the estimate available only
at the receiver) rather than the estimate of the channel. Since the
estimation error variance is computed offline, the adaptive
transmission scheme can then be determined offline as well and adopted
by the transmitter. Even though three is no feedback to the
transmitter, it is aware of the statistics of the estimation error
beforehand.

The transmitter employs multiple codebooks in an interleaved fashion
as shown in Figure~\ref{fig:figInter}.  It adapts its throughput to
the estimation error variance by coding the data symbols according to
their distance from the training pilots. Symbols that are far away
from the pilots encounter poorer channel estimates at the receiver and
are therefore coded with lower rate codes, while closer symbols
benefit from small estimation error variance and are coded with higher
rate codes.

\begin{figure}[!ht]
  \begin{center}
    \setlength{\unitlength}{0.78cm}
    \begin{picture}(11.5,4.55)(0,0)

  \put(0,0){\vector(1,0){11.5}}

  \put(0,0){\vector(0,1){1}}
  \put(0.5,0){\vector(0,1){1}}
  \put(1,0){\vector(0,1){1}}

  \put(1.5,0){\circle*{0.1}}
  \put(1.5,0.2){\makebox(0,0)[b]{$\$$}}
  \put(2,0){\circle*{0.1}}
  \put(2,0.2){\makebox(0,0)[b]{$\%$}}
  \put(2.5,0){\circle*{0.1}}
  \put(3,0){\circle*{0.1}}

  \put(3.5,0){\circle*{0.1}}
  \put(4,0){\circle*{0.1}}
  \put(4.5,0){\circle*{0.1}}
  \put(5,0){\circle*{0.1}}

  \put(5.5,0){\circle*{0.1}}
  \put(6,0){\circle*{0.1}}
  \put(6.5,0){\circle*{0.1}}
  \put(7,0){\circle*{0.1}}
  \put(7,0.2){\makebox(0,0)[b]{$\&$}}

  \put(7.5,0){\vector(0,1){1}}
  \put(8,0){\vector(0,1){1}}
  \put(8.5,0){\vector(0,1){1}}
  \put(9,0){\circle*{0.1}}
  \put(9,0.2){\makebox(0,0)[b]{$\$$}}
  \put(9.5,0){\circle*{0.1}}
  \put(9.5,0.2){\makebox(0,0)[b]{$\%$}}
  \put(10,0){\circle*{0.1}}
  \put(10.5,0){\circle*{0.1}}

  \put(0,4.3){\makebox(0,0)[bl]{\small Codebook 1}}
  \put(0,3.7){\makebox(0,0)[bl]{\small Rate $R_1$}}
  \put(0,2.5){\framebox(3,1)[l]{$\, \$ \, \$ \, \$ \, \cdots$ }}
  \put(4,4.3){\makebox(0,0)[bl]{\small Codebook 2}}
  \put(4,3.7){\makebox(0,0)[bl]{\small Rate $R_2$}}
  \put(4,2.5){\framebox(3,1)[l]{$\, \% \, \% \, \% \, \cdots$ }}
  \dashline{0.1}(7.5,3)(8,3)
  \put(8.5,4.3){\makebox(0,0)[bl]{\small Codebook $m$}}
  \put(8.5,3.7){\makebox(0,0)[bl]{\small Rate $R_m$}}
  \put(8.5,2.5){\framebox(3,1)[l]{$\, \& \, \& \, \& \, \cdots$ }}

  \dashline{0.15}(0.25,2.7)(1.4,0.6)
  \dashline{0.15}(0.55,2.7)(8.9,0.6)

  \dashline{0.15}(4.25,2.7)(2,0.6)
  \dashline{0.15}(4.8,2.7)(9.4,0.6)

  \dashline{0.15}(8.7,2.7)(7,0.6)

\end{picture}
    \caption{Multiple Codebook Interleaving
      \label{fig:figInter}}
  \end{center}
\end{figure}
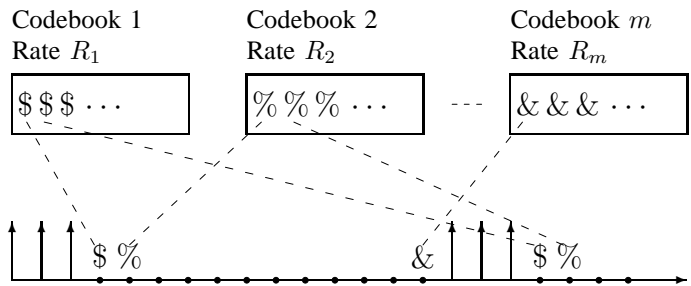

We only consider binary signaling. The motivation for this choice is
multiple folds. First, in~\cite{ATS01} the authors prove that for
discrete-time memoryless Rayleigh fading channels subject to average
power constraints, the capacity achieving distribution is discrete
with a finite number of mass points. Moreover, a binary distribution
was found to be optimal at low and moderate values of SNR~\cite{ATS01,
  Gall87, Verdu90}. Second, for a memoryless Rician fading channel,
Luo~\cite{Cheng06} established a similar result that, combined with
Gallager's in~\cite{Gall87}, implies that the binary input distribution
is asymptotically optimal at low SNR~\cite{Cheng06}.  Consequently, we
choose the alphabet of every codebook to consist in general of two
symbols:
\begin{equation*}
  \left \{ \begin{array}{llll}
      m_{1}& = a_{1}+jb_{1} & \text{\small with  
        probability} & p_1 \\ 
      m_{2}& = a_{2}+jb_{2} & \text{\small with 
        probability} & p_2 = (1-p_1).
    \end{array}\right. 
\end{equation*}

The rate of the codebooks is adjusted by modifying the probability
distribution of the mass points. Numerical results in~\cite{AMM05,
  ref4} indicate that the optimal mass points always lie between the
extremes of on-off keying (optimal for the IID Rayleigh fading case
where no CSI is available at the receiver) and the antipodal signaling
(optimal for a perfectly known channel). It is worth noting that some
of the work in the literature consider these two extremes for
designing the constellation mapping and try to optimize the
transmission model in the case of imperfect CSI based on the SNR
level~\cite {ref4}.  Moreover, any rotational transformation of the
two mass points will not affect the mutual information~\cite{ATS01,
  Cheng06}. Therefore an optimal input distribution consists of two
mass points $m_{1}, m_{2} \in\mathbb{R^*}$ with $-\sqrt{P}\leq m_1<0$
and $m_2\geq \sqrt{P}$.

\subsection{Channel Estimation at the Receiver}

Given a pilot spacing interval $T$, we send $k$ pilots in the
beginning of every data frame as shown in
Figure~\ref{fig:figKalEst}. When transmitting a pilot at time index
$i$, the input of the channel is $\sqrt{P_i}$ and its output is,
\begin{equation*}
  Y_{i} = \sqrt{P_i} R_{i} + N_{i}, \qquad i = 0, \cdots, k-1.
\end{equation*}

\begin{figure}[!ht]
  \begin{center}
    \setlength{\unitlength}{0.73cm}
    \begin{picture}(11.5,5)(0,-2.5)

  \put(0,0){\vector(1,0){11.5}}

  \put(0,0){\vector(0,1){2}}
  \put(0,-0.2){\makebox(0,0)[t]{0}}
  \put(0.5,0){\vector(0,1){1}}
  \put(0.5,-0.2){\makebox(0,0)[t]{1}}
  \dashline{0.1}(0.7,0.5)(1.3,0.5)
  \put(1.5,0){\vector(0,1){1.5}}

  \put(0,2.2){\makebox(0,0)[b]{$\sqrt{P_0}$}}
  \put(1.7,1.7){\makebox(0,0)[b]{$\sqrt{P_{k-1}}$}}

  \put(2,0){\circle*{0.1}}
  \put(2,-0.2){\makebox(0,0)[t]{$k$}}
  \put(2.5,0){\circle*{0.1}}
  \put(3,0){\circle*{0.1}}

  \put(3.5,0){\circle*{0.1}}
  \put(4,0){\circle*{0.1}}
  \put(4.5,0){\circle*{0.1}}
  \put(5,0){\circle*{0.1}}

  \put(5.5,0){\circle*{0.1}}
  \put(6,0){\circle*{0.1}}
  \put(6.5,0){\circle*{0.1}}
  \put(7,0){\circle*{0.1}}

  \put(7.5,0){\vector(0,1){2}}
  \put(7.5,-0.2){\makebox(0,0)[t]{$T$}}
  \put(8,0){\vector(0,1){1}}
  \dashline{0.1}(8.3,0.5)(8.8,0.5)
  \put(9,0){\vector(0,1){1.5}}


  \put(9.5,0){\circle*{0.1}}
  \put(10,0){\circle*{0.1}}
  \put(10.5,0){\circle*{0.1}}

  \put(4.5,1.2){\makebox(0,0)[lb]{$[\hat{R}_j, v_j]$}}
  \dashline{0.15}(4.7,1.1)(4,0.2)

  \put(2,-1){\vector(1,0){5.5}}
  \put(2,-1){\vector(-1,0){0}}
  \put(5.25,-1.15){\makebox(0,0)[t]{\small $T-k$ Data Symbols}}

  \put(0,-2){\vector(1,0){7.5}}
  \put(0,-2){\vector(-1,0){0}}
  \put(3.75,-2.15){\makebox(0,0)[t]{\small $T$ Pilot Spacing Interval}}

\end{picture}
    \caption{Pilot Symbols and Channel Estimation
      \label{fig:figKalEst}}
  \end{center}
\end{figure}
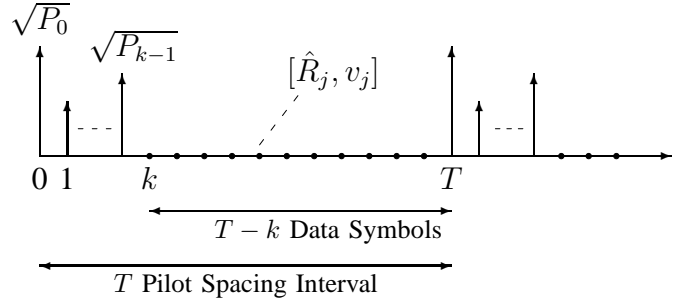

On the receiver side, we perform MMSE estimation based on 
the received signal during training. More precisely, we denote
by $\set{S}$ the set of indices corresponding to the received pilots
$\{Y_s\}_{s\in \set{S}}$ involved in estimating $R_{j}$ for $j = k,
\ldots, T-1$.  Therefore, when $\set{S} = \left\{ 0, \ldots,
  k-1\right\}$ we say we are performing causal MMSE estimation, and
when $\set{S} = \left\{ 0, \ldots, k-1,\, T, \ldots, T+k-1 \right\}$
the MMSE estimate is said to be non-causal.

Next, we compute the MMSE estimate $\hat{R}_j \left( \left\{ Y_{s}
  \right\}_{s\in \set{S}} \right)$ of $R_{j}$ for $j = k, \ldots,
T-1$. Since the random variables $\left\{ R_{j},\,\,\{Y_{s}\}_{s\in
    \set{S}} \right\} $ are jointly Gaussian, the MMSE estimator is
linear and is identical to the Linear Least-Square Estimator (LLSE)
the error variance $v_j$ of which is,
\begin{equation}
  {\displaystyle v_j = 1 - \Lambda_{ R_{j},\,\{Y_s\}_{s\in \set{S}} } 
    \; \Lambda_{\{Y_s\}_{s\in \set{S}} }^{-1}
    \; \Lambda_{ R_{j},\,\{Y_s\}_{s\in \set{S}} }^T\,\,,}
  \label{eqn:error}
\end{equation}
where $\Lambda_{R_{j},\,\{Y_s\}_{s\in \set{S}}}$ is the
cross-covariance matrix between $R_{j}$ and $\{Y_s\}_{s\in \set{S}}$
and $\Lambda_{\{Y_s\}_{s \in \set{S}} }$ is the autocovariance of the
vector of received pilots $\{Y_s\}_{s\in \set{S}}$.

We note that the estimation error variance in
equation~(\ref{eqn:error}) may be computed offline at design time
--and therefore no feedback is needed to the encoder-- and is only
dependent on the autocorrelation function of $\{R_j\}$, the
transmitted pilots and the noise spectral density.

\section{Achievable Rates}
\label{sc:MI}

We consider the transmission scheme shown in
Figure~\ref{fig:figKalEst} with symbols sent with power $P_j$ for $j =
0, \ldots, T-1$. Given a sample $\left\{ y_{s}\right\}_{s\in
  \set{S}}$, the received symbols can be written
\begin{equation*}
  Y_{i} = R_{i}X_{i} + N_{i}
  = \left( \hat{R}_{i} + \Gamma_{i} \right) X_{i} + N_{i},
  \quad \text{for } i = k, \ldots, T-1,
\end{equation*}
where $\Gamma_i$ is a zero-mean complex Gaussian error term that has a
variance $v_i$. Therefore,
\begin{eqnarray*}
  p \left( y_{i}|x_{i},\left\{y_{s}\right\}_{s\in \set{S}} \right)
  & = & \mathcal{N}_{\mathbb{C}} \left( \hat{R}_{i}x_{i}, \, v_{i}
    \left\vert x_{i} \right\vert^{2} + \sigma_{N}^{2}\right)  
  = \frac{1}{\pi \left( v_{i}\left\vert x_{i}\right\vert ^{2}
      +\sigma_{N}^{2}\right) } e^{ -\frac{\left\vert y_{i}-
        \hat{R}_{i}x_{i}\right\vert ^{2}}{v_{i}\left\vert x_{i}
      \right\vert ^{2}+\sigma_{N}^{2}}}.
\end{eqnarray*}

When ignoring the fading correlation from one transmitted frame to
another, the mutual information per symbol due to interleaving can be
written as
\begin{multline}
	\frac{1}{T} \, I \left( \; \{X_i\}_{i=0}^{T-1} \; ; \; \{Y_i\}_{i=0}^{T-1} \;
	| \; \{Y_s\}_{s\in \set{S}} \; \right)
  = \Ep{\{Y_s\}_{s\in \set{S}}} {\frac{1}{T}\sum_{i=k}^{T-1}
    I \left( X_{i};Y_{i}|\{y_s\}_{s\in \set{S}} \right)} \\
     = \frac{1}{T}\sum_{i=k}^{T-1} \Ep{\hat{R}_{i}} {
    I \left( X_{i};Y_{i}|\hat{R}_{i}\right) }
  \label{eqn:lastexp},
\end{multline}
where the expectation is now over the random variable
$\hat{R}_{i}$. Note that $\hat{R}_i$ is a linear combination of the
observations
\begin{equation}
  \hat{R}_{i}=\sum_{m\in\set{S}}\beta _{m}Y_{m}=R_{i}-
  \Gamma _{i}\sim\mathcal{N}_{\mathbb{C}}(0,1-v_{i}).
  \label{eq:pdf_Rhat}
\end{equation}

\subsection{The Computation Method}

The $i^{th}$ term, $\Ep{\hat{R}_{i}}{I\left( X_{i};Y_{i}|
    \hat{R}_{i}\right) }$, in equation~(\ref{eqn:lastexp}) depends on
the choice of the corresponding binary probability distribution fully
characterized by the three parameters $\{m_1, m_2, p_1\}_i$. This
distribution (for $i = k, \ldots ,T-1$) determines the rate of the
corresponding codebook and should be chosen to maximize the mutual
information quantity in~(\ref{eqn:lastexp}). Therefore, we are
interested in solving
\begin{equation}
  \frac{1}{T} \sum_{i=k}^{T-1} \max_{\{m_1,\,m_2,\,p\}_i}
  \Ep{\hat{R}_{i}} { I\left( X_{i};Y_{i} | \hat{R}_{i} \right) }
  \label{eqn:optimObj},
\end{equation}
subject to $\E{\vert X_i\vert^2} \leq P_i$ for all $i = k, \ldots,
T-1$.

Furthermore, examining the probability law~(\ref{eq:pdf_Rhat}) of
$\hat{R}_{i}$ indicates that the elementary quantity $\displaystyle{
  \max_{\{m_1,\,m_2,\,p\}_i} \Ep{\hat{R}_{i}} { I \left( X_{i};Y_{i} |
      \hat{R}_{i}\right) }}$ in~(\ref{eqn:optimObj}) is only a
function of the estimation error variance $v_i$ and power $P_i$ of the
symbol. We define
\begin{equation*}
  I_{sub}(P_i,\,\,v_i) =\max_{\{m_1,\,m_2,\,p\}_i}\Ep{
    \hat{R}_{i}} {I\left(X_{i};Y_{i}|\hat{R}_{i}\right)},
\end{equation*}
where the maximization is subject to $\E{\vert X_i\vert^2} \leq P_i$
and $\hat{R_i}\sim \mathcal{N}_{\mathbb{C}} (0, 1 -
v_{i})$. Thereafter the achievable rates become
\begin{equation}
  \frac{1}{T} \, I \left( \; \{X_i\}_{i=0}^{T-1} \; ; \; \{Y_i\}_{i=0}^{T-1} \;
	| \; \{Y_s\}_{s\in \set{S}} \; \right)
  = \frac{1}{T} \sum_{i=k}^{T-1} I_{sub}(P_i,\,\,v_i)
  \label{eqn:optimObj2}.
\end{equation}

The two dimensional curve $ I_{sub}(P,\,\,v) $ is computed over a fine
grid $\mathscr{V} = \{0 \leq P \leq P_{max}, 0\leq v\leq 1\}$ as shown
in Figure~\ref{fig:Isub}.  Then given a transmission strategy
consisting of an inter-pilot spacing $T$, $k$-pilot clustering, and a
power allocation $P_j$ for $j= 0, \ldots ,T-1$, we calculate using
equation~(\ref{eqn:error}) the estimation error variance $v_j$ for $j
=k,\,\dots, T-1$. The corresponding elementary mutual information
quantity $I_{sub}(P_j,\,v_j)$ can now be interpolated from the data
set $\{\mathscr{V},\,I_{sub}(P,\,v)\}$ and used to compute the
normalized sum in (\ref{eqn:optimObj2}).

\begin{figure}[!hb]
  \begin{center}
    \includegraphics[width=3.7in]{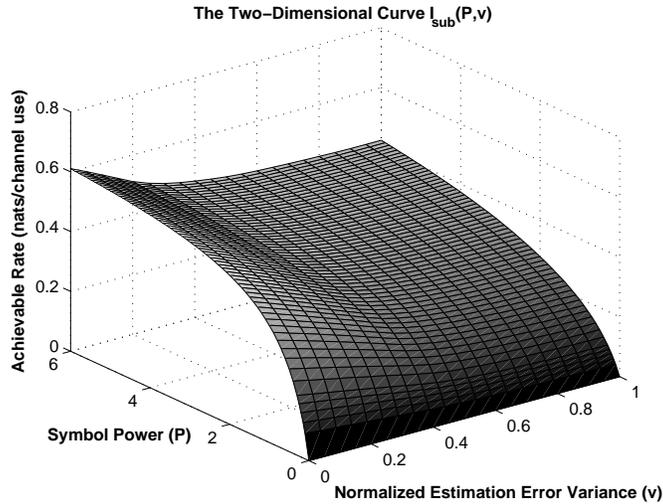}
    \caption{\small The two-dimensional curve $ I_{sub}(P,\,\,v) $
      \label{fig:Isub}}
  \end{center}
\end{figure}

Finally, note that the error variance is a function of the power of
the pilots $\{P_{s}\}_{s\in \set{S}}$. Hence
equation~(\ref{eqn:optimObj2}) can also be written as
\begin{equation}
  I\left( \; \{X_i\}_{i=k}^{T-1} \; ; \; \{Y_i\}_{i=k}^{T-1} \;
    | \; \{Y_s\}_{s\in \set{S}} \; \right)
  = \frac{1}{T} \sum_{i=k}^{T-1} I_{sub} \left( P_i,\,\, \{P_{s}\}_{s\in \set{S}} \right)
  \label{eqn:optimObj3}.
\end{equation}

\subsection{The transmission policy}

We consider four types of transmission policies and we study how the
optimal training strategy differs from one policy to another,
analytically in Section~\ref{sc:OPT} and numerically in
Section~\ref{sc:NR}.

\subsubsection{Policy I} 

The pilot symbols and the data symbols are transmitted with the same
amount of power, i.e.,
\begin{equation*}
  P_s = P, \quad \forall s = 0, \ldots, k-1
  \qquad \& \qquad P_i = P, \quad \forall \, i=k,\ldots,T-1.
\end{equation*}

Therefore. for a given channel model, $k$-pilot training, and an
inter-pilot spacing $T$, the achievable rate in
equation~(\ref{eqn:optimObj3}) is a function of $P$ only.

\bigskip

\subsubsection{Policy II} 
\label{sc:policy2}

In this policy, a flat power allocation is adopted for both the pilot
symbols and the data symbols, but we allow the two levels to be
different. More precisely,
\begin{equation*}
  P_s = P_{tr} \quad \forall s = 0, \ldots, k-1
  \qquad \& \qquad P_i = P_d, \quad \forall \, i=k,\ldots,T-1.
\end{equation*}

The achievable rate is a function of $P_{tr}$ \& $P_d$ which satisfy
\begin{equation*}
  \frac{1}{T}\sum_{j=0}^{T-1} P_j = \frac{1}{T} \bigl[ k P_{tr} + (T-k) P_d \bigr] \leq P.
\end{equation*}  

\bigskip

\subsubsection{Policy III}

Following a flat power allocation for pilots ($P_{s} = P_{tr}, \forall
s = 0, \ldots, k-1$), the data symbols are sent with power $P_i$ for
$i = k, \ldots, T-1$. These power levels satisfy
\begin{equation*}
  \frac{1}{T}\sum_{j=0}^{T-1}P_j = \frac{1}{T} \left[k   P_{tr} + \sum_{j=k}^{T-1}P_j \right] \leq P.
\end{equation*}

\bigskip

\subsubsection{Policy IV}
\label{sc:policy4}

We send both the pilots and data symbols with variable power $P_j$ for
$j = 0,\ldots,T-1$. The constraint on the power levels is now given by
\begin{equation*}
  \frac{1}{T}\sum_{j=0}^{T-1}P_j\leq P.
\end{equation*}

\section{Optimal Power Allocation for Causal Estimation}
\label{sc:OPT}

In this section we find the optimal power allocation and training
duration for policies II, III and IV under causal estimation. These
optimal solutions are found by applying the result of
Theorem~\ref{th:PA_PIV} stated hereafter. The theorem implies that if
we let $k P_{tr}$ be the {\em total\/} power ``budget'' for the
training period, everything else being equal, among all the training
power allocation schemes $\{ P_s \}_{s=0}^{k-1}$ such that
\begin{equation}
  \sum_{s=0}^{k-1} P_s = k P_{tr},
  \label{eq:powtr}
\end{equation}
the optimal one is the one where all the power is allocated to the
last time slot $(k-1)$.  

For causal MMSE estimation $\set{S} = \{0, \cdots, k - 1\}$ and
equation~(\ref{eqn:error}) can be written as
\begin{align}
  v_j & = 1 - \Lambda_{ R_{j},\,\{Y_s\}_{s\in \set{S}} } 
  \; \Lambda_{\{Y_s\}_{s\in \set{S}} }^{-1} \; \Lambda_{ R_{j},\,\{Y_s\}_{s\in \set{S}} }^T 
  \qquad \qquad \qquad j = k, \cdots, (T-1) \nonumber \\
  & = 1 - \alpha^{2(j-k+1)} \left( \alpha^{k-1} \,\, \alpha^{k-2} \,\, \cdots \,\, 1
  \right) \; D^{T} \; \left[ D \; A \; D^{T} + \sigma_N^{2}\;I \right]^{-1}
  \; D \; \left( \begin{array}{c} \alpha^{k-1} \\ 
      \vdots \\ 1 \end{array} \right), \label{eq:nerr}
\end{align}
where $A$ is the $k \times k$ symmetric, positive definite
autocovariance matrix of the channel fading coefficients
$\{R_s\}_{s\in \set{S}}$, and $D$ is the $k \times k$ ``input''
matrix:
\begin{equation*}
  A = \begin{pmatrix}
    1 & \alpha & \cdots & \alpha^{k-1} \\
    \alpha & 1 & \cdots & \alpha^{k-2} \\
    \vdots & \vdots & \ddots & \vdots \\
    \alpha^{k-1} & \alpha^{k-2} & \cdots & 1
  \end{pmatrix},
  \qquad 
  D = \begin{pmatrix}
    \sqrt{P_{0}} & 0 & \cdots & 0  \\
    0 & \sqrt{P_{1}} & \cdots & 0  \\
    \vdots & \vdots & \ddots & \vdots  \\
    0 & 0 & \cdots & \sqrt{P_{k-1}} 
  \end{pmatrix}.
\end{equation*}

A power allocation that minimizes the error variances of the estimates
for all $\{j\}$'s --subject to the power constraint~(\ref{eq:powtr})--
is naturally an optimal one. Examining~(\ref{eq:nerr}), we note that a
power allocation that minimizes $v_{j_o}$ for some $j_o$ will also
minimize $v_j$ for all $\{j\}$'s, as it will be one that maximizes
\begin{equation}
  \left( \alpha^{k-1} \,\, \alpha^{k-2} \,\, \cdots \,\, 1
  \right) \; D^{T} \; \left[ D \; A \; D^{T} + \sigma_N^{2} \;I \right]^{-1}
  \; D \; \left( \begin{array}{c} \alpha^{k-1} \\
      \vdots \\ 1 \end{array} \right).
  \label{eq:obj}
\end{equation}

The power allocation that maximizes~(\ref{eq:obj}) is the subject of
the following theorem, proven in Appendix~\ref{ap:proof}.
\begin{theorem}
  \label{th:PA_PIV}
  The expression~(\ref{eq:obj}) is maximized when all the available
  power is allocated to the last pilot, i.e., $P_j = 0$, for all $0
  \le j \le (k-2)$ and $P_{k-1} = kP_{tr}$.
\end{theorem}

We note that Theorem~\ref{th:PA_PIV} holds whenever one allows the
power allocation during the training period to vary across the
pilots. We also note that the result holds irrespective of how the
power is allocated for the data.

\subsubsection*{Implications on the training duration} 
\begin{itemize}
\item When considering policy IV, the powers of the individual
  training symbols are allowed to vary and the theorem states that all
  the power should be allocated to the last training symbol. Factoring
  in the loss of achievable rates due to training, it becomes clear
  that the optimal duration is that of {\em one\/} pilot transmission.
\item Since the achievable rates using policy III are less or equal to
  those of policy IV, and since the optimal solution for policy IV is
  that of a ``flat'' power allocation over the duration of the
  training --which is one, then the solution is also optimal for
  policy III.
\item Finally, since the statement of the theorem is valid irrespective
  of how the power is allocated during data transmission and
  specifically even when a flat power allocation is used, the result
  implies that for policy II, using a training duration of one pilot
  is optimal as well.
\end{itemize}

Naturally, these statements are true if the power level during
training is optimized. In Section~\ref{sc:NR} we validate numerically
these results.

\section{Numerical Results}
\label{sc:NR}

For a given channel model, a given SNR (power constraints), and
estimation technique (causal or non-causal), we numerically determine
the optimal training strategy consisting of:
\begin{enumerate}[1.]
\item The duration of training or the number of pilots $k$.
\item The inter-pilot spacing $T$.
\item The power allocation for the pilots and data symbols in a
  transmitted frame, according to the transmission policy used.
\end{enumerate}

In our work, the quality measure is the achievable rates which we
compute for pilot clustering/training period of up to six pilots in
each frame. We study the low \emph{received} SNR regime (SNR values of
-3dB, 0dB, 3dB, and 6dB) for a first-order Gauss-Markov fading process
with values of $\alpha = 0.9, 0.95, 0.97,\,\textrm{and}\,0.99$.  On
the receiver side, causal and non-causal estimation are
investigated. We present hereafter graphs for some chosen test cases
and compare the rates achieved using 1, 2, 3, 4, 5, and 6-pilot
clustering strategies for different scenarios of SNR and fading
correlation levels.

We note first that the numerical results confirm the observation
previously made that the achievable rate in
equation~(\ref{eqn:lastexp}) depends on the choice of $\{m_1, m_2,
p_1\}_i$, i.e., the input distribution of the $i$-th symbol. As the
symbol gets further away from the training pilots, the channel
estimation quality (measured through the estimate error variance) is
degraded and hence the amount of information sent over the channel
decreases. This is translated by shifting $\{m_1, m_2, p_1\}_i$ from
the antipodal distribution (optimal for a perfectly known channel)
with $p_1 \approx p_2$ (high entropy) toward the other extreme of
on-off keying (optimal for the IID Rayleigh fading case) with $p_1\gg
p_2$ (low entropy).

We also note that in the case of causal estimation, our numerical
results are consistent with the results in Section~\ref{sc:OPT}.

\subsection{Results for Transmission Policy I}

For transmission policy I, pilot clustering proves to achieve higher
rates under certain conditions compared to the 1-pilot scheme.  In
Figure~\ref{fig:gcf1}, for an SNR = 0dB, $\alpha$ = 0.99 and causal
estimation, training with 4 pilots and inter-pilot spacing of $T$=29
symbols is optimal. A percent increase of 8.2\% in information rate is
achieved relative to the best rate achievable with a 1-pilot
scheme. The results for other test cases are shown in
Table~\ref{table:table}.  

\begin{figure}[!hp]
  \begin{center}
    \includegraphics[width=3.7in]{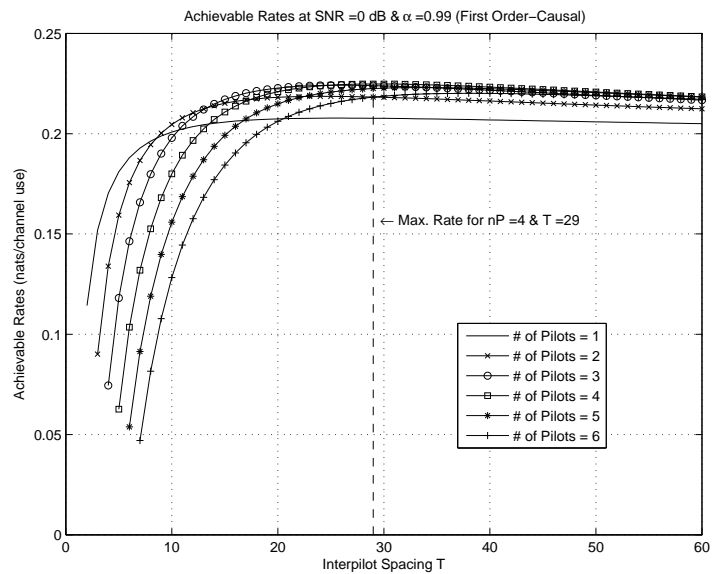}
    \caption{\small Achievable Rates for policy I, for SNR = 0dB and
      $\alpha$ = 0.99 with Causal Estimation.
      \label{fig:gcf1}}
  \end{center}
\end{figure}

\begin{table}
  \begin{center}
    \begin{threeparttable}
      \caption{Achievable Rates for Different Transmission Policies}
      \centering
      \begin{tabular}{|c||c|c|c|c|}
        \hline
        \hline
        Test Case & Policy I  & Policy II & Policy III & Policy IV \\ \hline
        \hline
        
        {$\alpha\, =\, 0.9$} & nP=4 & nP=1 & nP=1 & nP=1 \\
        {SNR = 0dB} & T=29 & T=22 & T=23 & T=22  \\
        {Causal Estimation} & Rate$\approx$0.2247 & Rate$\approx$0.2418 & 
        Rate$\approx$0.2422 & Rate$\approx$0.2422 \\
        &8.2\%$\uparrow$\tnote{1} & 7.6\%$\uparrow$\tnote{2}& & \\\hline
        
        {$\alpha\, =\, 0.97$} & nP=1 & nP=1 & nP=1 & nP=1 \\
        {SNR = 6dB} & T=15 & T=15 & T=15 & T=15  \\
        {Causal Estimation} & Rate$\approx$0.3782 & Rate$\approx$0.3829 & 
        Rate$\approx$0.3836 & Rate$\approx$0.3836 \\
        & &1.2\%$\uparrow$\tnote{2}& &\\\hline
        
        {$\alpha\, =\, 0.97$} & nP=3 & nP=1 & nP=1 & nP=1 \\
        {SNR = -3dB} & T=19 & T=18 & T=18 & T=18  \\
        {Non-Causal Estimation} & Rate$\approx$0.1374 & Rate$\approx$0.1470 & 
        Rate$\approx$0.1472 & Rate$\approx$0.1472 \\
        &4.3\%$\uparrow$\tnote{1} & 6.9\%$\uparrow$\tnote{2}& & \\\hline \hline
        
        {Using Jakes' model:} & & & & \\
        {$f_d\, =\,100\,$Hz, $f_s\,=\,10\,$KHz} & nP=2 & nP=1 & nP=1 & nP=1 \\
        {SNR = 3dB} & T=14 & T=13 & T=13 & T=13  \\
        {Causal Estimation} & Rate$\approx$0.3224 & Rate$\approx$0.3508 & 
        Rate$\approx$0.3510 & Rate$\approx$0.3510 \\
        &4\%$\uparrow$\tnote{1} & 8.8\%$\uparrow$\tnote{2}& & \\\hline
        
        {Using Jakes' model:} & & & & \\
        {$f_d\, =\,100\,$Hz, $f_s\,=\,10\,$KHz} & nP=4 & nP=1 & nP=1 & nP=1 \\
        {SNR = 0dB} & T=29 & T=30 & T=30 & T=30  \\
        {Non-Causal Estimation} & Rate$\approx$0.2843 & Rate$\approx$0.3064 & 
        Rate$\approx$0.3064 & Rate$\approx$0.3064 \\
        &16\%$\uparrow$\tnote{1} & 7.7\%$\uparrow$\tnote{2}& & \\\hline
      \end{tabular}\label{table:table}
      \begin{tablenotes}
      \item[1] relative to the rate achieved by the 1-pilot scheme (Policy I).
      \item[2] relative to the achievable rate under Policy I.
      \end{tablenotes}
    \end{threeparttable}
  \end{center}
\end{table}

However there are some scenarios when pilot-clustering is not useful.
For the case when SNR = 6dB, $\alpha$ = 0.97 and causal estimation,
the 1-pilot scheme presents optimal rates.  

Moreover in Figure~\ref{fig:gcf2} at an SNR=0dB, and $\alpha$=0.9 with
causal estimation, training is not beneficial in the first place
because the information rate is less than that achieved over an IID
Rayleigh fading channel.

\begin{figure}[!htp]
  \begin{center}
    \includegraphics[width=3.7in]{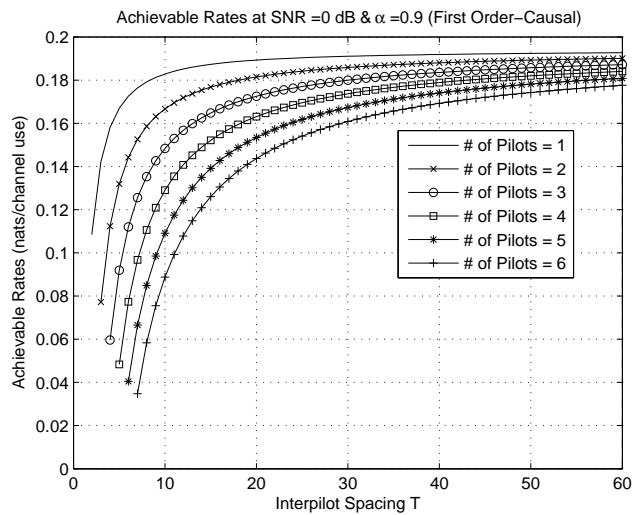}
    \caption{\small Achievable Rates for policy I, for SNR = 0dB and
      $\alpha$ = 0.9 with Causal Estimation.
      \label{fig:gcf2}}
  \end{center}
\end{figure}

As a conclusion, we can distinguish three cases. The first is when
training is not applicable. The second is when the 1-pilot scheme
gives the highest rates. And finally the third when pilot clustering
is beneficial. From our numerical results, we note that as SNR
increases and coherence time decreases, clustering becomes useless and
the whole scheme is pushed toward the 1-pilot training strategy and
even to the extreme case of no training at all.  This is directly
related to the fact that training is inefficient (less CSI) when
fading decorrelates quickly or when SNR is high.


\subsection{Results for Transmission Policy II}
\label{sc:policyII}

In this policy, the pilots are sent with fixed power $P_{tr}$ ({\em
  $\,$per pilot$\,$}) and so are the symbols that are transmitted with
power $P_d$ (Section~\ref{sc:policy2}), such that
$\displaystyle{\frac{1}{T}\sum_{j=0}^{T-1}P_j\leq P}$. Therefore, the
optimal training strategy includes determining the optimal power
allocation ($P_{tr}$ and $P_d$) for the transmitted frame. Here the
notion of SNR is naturally associated with the average power $P$.

Figure~\ref{fig:gcf3} shows the achievable rates for SNR=0dB, and
$\alpha$=0.99 with causal estimation.  Unlike the results for policy I
(Figure~\ref{fig:gcf1}), training with 4 pilots is not optimal
anymore.  The 1-pilot scheme (with $T$=22) now offers 7.6\% increase
in the achievable rate compared to the 4-pilot scheme for policy
I. The corresponding optimal power allocation across the transmission
frame is shown in Figure~\ref{fig:gcf4}.

\begin{figure}[!htp]
  \begin{center}
    \includegraphics[width=3.7in]{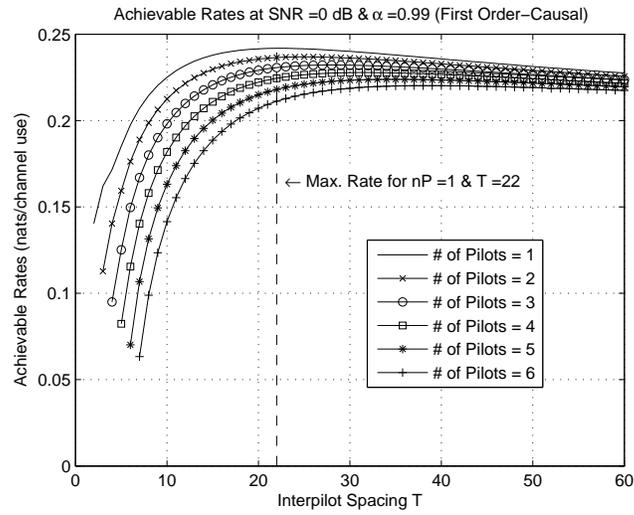}
    \caption{\small Achievable Rates for policy II, for SNR = 0dB and
      $\alpha$ = 0.99 with Causal Estimation.
      \label{fig:gcf3}}
  \end{center}
\end{figure}

\begin{figure}[!htp]
  \begin{center}
    \includegraphics[width=3.7in]{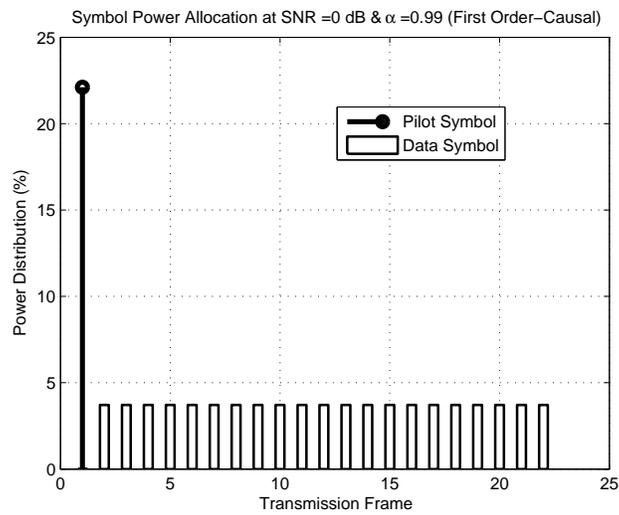}
    \caption{\small Optimal Symbol Power Allocation (one frame) for
      policy II, for SNR = 0dB and $\alpha$ = 0.99 with Causal
      Estimation.
      \label{fig:gcf4}}
  \end{center}
\end{figure}

The rest of the results are presented in Table~\ref{table:table} and
they all confirm that, as expected pilot clustering is not optimal for
policy II, and for any transmission strategy where the pilots' power
is subject to optimization for that matter.  In this case, the
transmitter decreases the estimation error variance (higher
throughput) by boosting the power of the single pilot instead of
increasing the number of pilots $k$ and getting penalized by the
normalizing term $\displaystyle{\frac{1}{T}}$ in
equation~(\ref{eqn:optimObj3}). 

If a peak power constraint is imposed on the power of the pilots, the
optimal training duration will not necessarily be one pilot. This can
be seen from Figure~\ref{fig:gcf4_p} which shows the optimal power
allocation across the transmission frame for SNR=0dB, and
$\alpha$=0.99 with causal estimation whenever a peak constraint
$P_{tr} \leq 3 P$ is imposed. This constraint is effectively imposing
a maximum Peak-to-Average Power Ratio (PAPR) value of 3.

\begin{figure}[!htp]
  \begin{center}
    \includegraphics[width=3.7in]{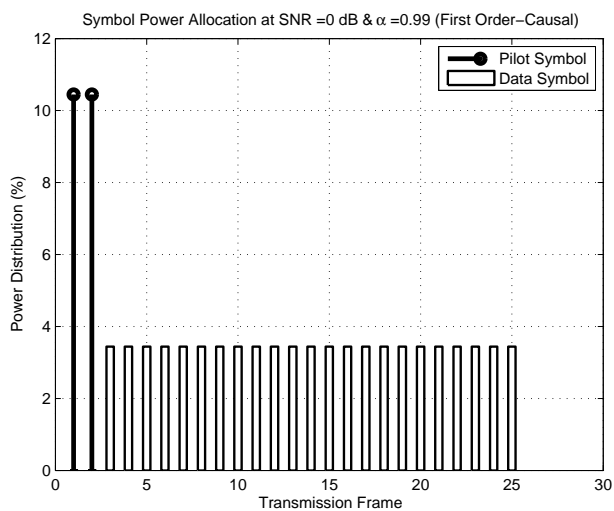}
    \caption{\small Optimal Symbol Power Allocation (one frame) for
      policy II, for SNR = 0dB and $\alpha$ = 0.99 with Causal
      Estimation and peak constraint $P_{tr} \leq 3 P$.
      \label{fig:gcf4_p}}
  \end{center}
\end{figure}

As mentioned earlier, there are some scenarios where training is not
useful and the rate is always less than that achieved over an IID
Rayleigh fading channel. This is observed with causal estimation for
an SNR=0dB and $\alpha$=0.9 for example. In that case all the power is
allocated to the data symbols indicating that training is not
beneficial.


\subsection{Results for Transmission Policy III}

For policy III, we send the data symbols with varying power as we hold
on to a flat power allocation for the pilots.  As already shown in the
Section~\ref{sc:OPT}, clustering is not useful for this case as well.
The transmitter boosts the power of the single pilot used in training
to decrease the error variance and increase the achievable rate.

The numerical results are in accordance with those of
Section~\ref{sc:OPT} and they show how the power of the symbols is
adapted to the estimation error variance. In Figure~\ref{fig:gcf6},
the power allocated to each symbol and the variation of the error
variance are presented for an SNR=0dB, and $\alpha$=0.99 with
non-causal estimation.  This shows that symbols with lower variance
are sent with higher power and vice versa.  However we should note
that power variations among the data symbols is not profound.

\begin{figure}[!htp]
  \begin{center}
    \includegraphics[width=3.7in]{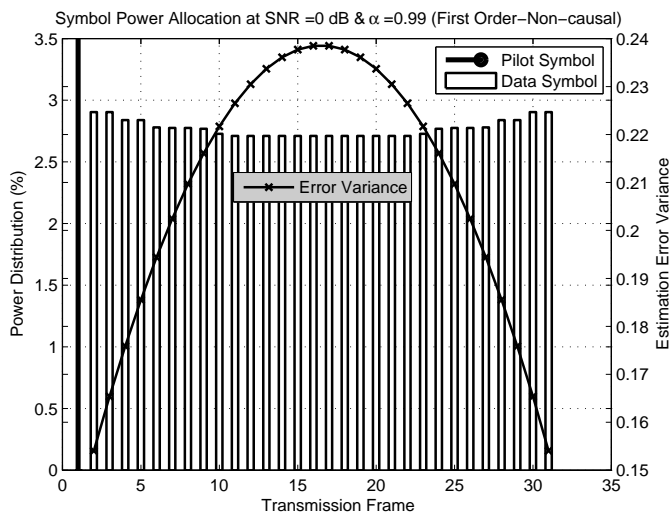}
    \caption{\small Optimal Symbol Power Allocation (one frame) for
      policy III, for SNR = 0dB and $\alpha$ = 0.99 with Non-Causal
      Estimation.
      \label{fig:gcf6}}
  \end{center}
\end{figure}

The achievable rates for other cases are summarized in
Table~\ref{table:table}. It is noticed that adapting the symbol power
to the quality of estimation introduces a slight increase in
achievable rates compared to policy II. As a result, one can say that
uniform power allocation for the data symbols is sufficiently close to
optimal and presents a more practical transmission strategy.

\subsection{Results for Transmission Policy IV}


Here both the pilots and data symbols are sent with varying power
(Section~\ref{sc:policy4}). However from the results for transmission
policy III, we already know that sending the data symbols with uniform
power is very close to optimal.

Let us consider the case for an SNR=0dB, and $\alpha$=0.99. We choose
a 4-pilot training scheme.  For causal estimation, the power allocated
to the pilots is shown in Figure~\ref{fig:gcf7}. We notice that all of
the power was found numerically to be allocated to the pilot closest
to the symbols leaving the rest of the pilots that are further away
with no power and therefore useless, which is consistent with the
results of Theorem~\ref{th:PA_PIV} and Section~\ref{sc:OPT}.
Combining this result with the penalty factor
$\displaystyle{\frac{1}{T}}$ in equation~(\ref{eqn:lastexp}), we reach
the conclusion that the 1-pilot scheme is always optimal
(Table~\ref{table:table}).

\begin{figure}[!htp]
  \begin{center}
    \includegraphics[width=3.7in]{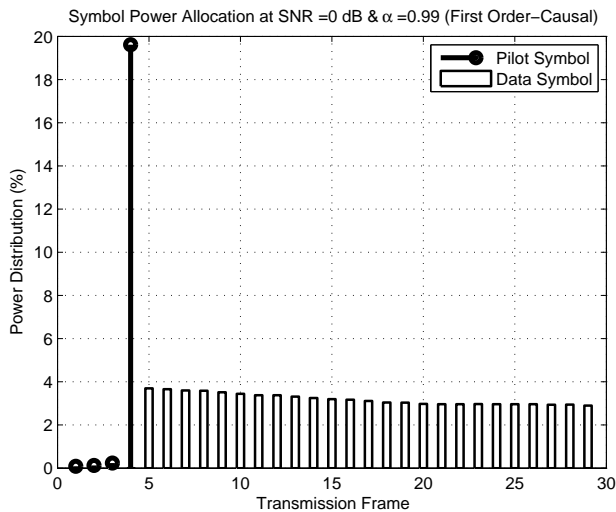}
    \caption{\small Optimal Symbol Power Allocation (one frame) for
      policy IV, for SNR = 0dB and $\alpha$ = 0.99 with Causal
      Estimation.
      \label{fig:gcf7}}
  \end{center}
\end{figure}

A similar result is shown in Figure~\ref{fig:gcf8} for the non-causal
estimation scenario. The powers of the first pilot (playing a
prominent role in the non-causal part) and last pilot (with a
prominent role in the causal part) are increased. 

\begin{figure}[!htp]
  \begin{center}
    \includegraphics[width=3.7in]{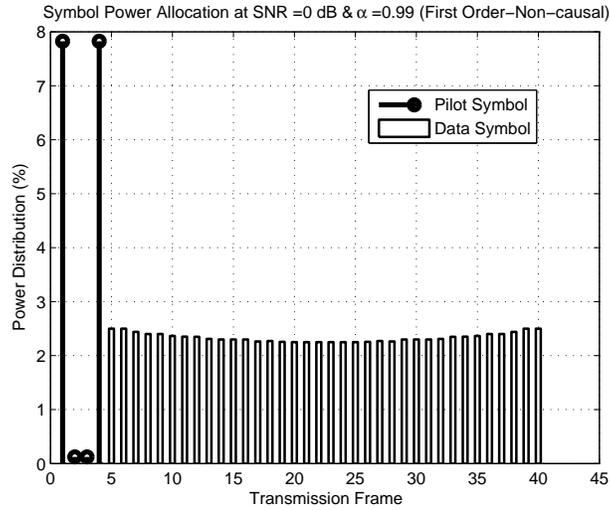}
    \caption{\small Optimal Symbol Power Allocation (one frame) for
      policy IV, for SNR = 0dB and $\alpha$ = 0.99 with Non-Causal
      Estimation.
      \label{fig:gcf8}}
  \end{center}
\end{figure}

Whenever a peak power constraint is imposed on the power of the
pilots, the optimal training duration will potentially involve pilot
clustering. The optimal duration and power allocation in
Figure~\ref{fig:gcf7_p} are for an SNR=0dB, and $\alpha$=0.99 with
causal estimation whenever a peak constraint $P_{tr} \leq 3 P$ is
imposed.

\begin{figure}[!htp]
  \begin{center}
    \includegraphics[width=3.7in]{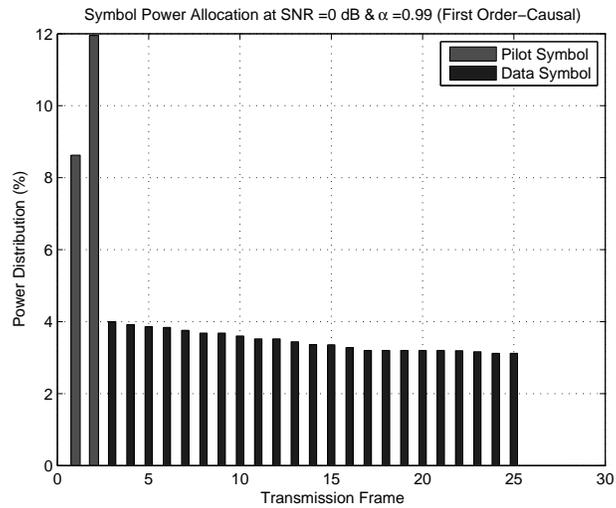}
    \caption{\small Optimal Symbol Power Allocation (one frame) for
      policy IV, for SNR = 0dB and $\alpha$ = 0.99 with Causal
      Estimation and peak constraint $P_{tr} \leq 3 P$.
      \label{fig:gcf7_p}}
  \end{center}
\end{figure}

\section{Other Fading Process Models}
\label{sc:Ext}

Whenever the fading process follows a different model, appropriate
results may be readily derived as the numerical optimization is only
dependent on the autocovariance function of the process as seen from
equation~(\ref{eqn:error}). In what follows, we present sample results
using Jakes' model.

\subsection*{Jakes' Model}

In Jakes' model~\cite{JakesBook}, the normalized (unit variance)
continuous-time autocorrelation function of the fading process is
given by
\begin{equation*}
  \phi_{RR}(\tau)=J_0(2\pi f_d \tau),
\end{equation*}
where $J_0(.)$ is the zeroth-order Bessel function of the first kind
and $f_d$ is the maximum Doppler frequency.  For the purposes of
discrete-time simulation of this model \cite{Baddour05}, the
autocorrelation sequence becomes
\begin{equation*}
  \phi_{RR}[l] = J_0(2\pi f_d\,T_s\,|l|),
\end{equation*}
where $1/T_s$ is the symbol rate.

In Table~\ref{table:table} we list a sample of the results obtained
for a bandwidth $f_s=10$ kHz and a Doppler shift of $f_d=100$ Hz. For
example, optimal training consists of $k=4$ and $T=29$ when we have an
SNR=0dB and non-causal estimation. Throughput is improved by 16\% in
this case.

\section{Conclusion}
\label{sc:conc}

We studied the performance of the non-feedback pilot-based adaptive
modulation scheme~\cite{AMM05, ZAF09, BAM04} over time-varying
Rayleigh fading channels. We measured the performance in terms of
achievable rates using binary signaling and we investigated the
benefits of pilot clustering as well as power allocation.


We introduced a modular method to compute the rates in an efficient
manner. Moreover, four types of transmission policies were
analyzed. For each policy, we determined the optimal training strategy
consisting of:
\begin{enumerate}[1.]
\item The duration of training.
\item The inter-pilot spacing.
\item The power allocation for the pilots and data symbols in the
  frame.
\end{enumerate}

Pilot clustering proved to be useful in the {\em low SNR--high
  coherence time} range where training is efficient (Policy
I). However, when the pilot power is subject to optimization (Policies
II, III and IV), training for a smaller period but with boosted power
becomes more beneficial than training with more pilots. We proved that
the optimal training duration using causal estimation is indeed one
whenever the power level during training is optimized and allowed to
take arbitrary values. Numerical results suggest that this is also the
case when using non-causal estimation at the receiver.

We also noted that the numerical computations indicate that a flat
power allocation across the data slots in a frame is very close to
optimal whenever the pilot power is subject to optimization.

On the other hand, training is useless in the {\em high SNR--small
  coherence time} range and the rate is always less than that achieved
over an IID Rayleigh fading channel. Several test cases are shown
throughout this work to analyze how optimal training varies with
channel conditions and from one transmission policy to another.

Extensions to this work can include adaptive schemes that integrate
temporal and spatial components like the Multiple-Input
Multiple-Output (MIMO) scenario.

\appendices
\section{}
\label{ap:proof}
 
In this appendix we provide a proof for Theorem~\ref{th:PA_PIV}. For
notational convenience, define
\begin{align*}
  \phi & \eqdef  V^{T} \; D^{T} \; \left[ D \; A \; D^{T} + \sigma_N^{2} \;I
  \right]^{-1} \; D \; V,
\end{align*}
where $V = \left( \alpha^{k-1} \,\, \alpha^{k-2} \,\, \cdots \,\, 1
\right)^T$, $A$ is the $k \times k$ symmetric, positive definite
autocovariance matrix of the channel fading coefficients
$\{R_s\}_{s\in \set{S}}$, and $D$ is the $k \times k$ ``input''
matrix:
\begin{equation*}
  A = \begin{pmatrix}
    1 & \alpha & \cdots & \alpha^{k-1} \\
    \alpha & 1 & \cdots & \alpha^{k-2} \\
    \vdots & \vdots & \ddots & \vdots \\
    \alpha^{k-1} & \alpha^{k-2} & \cdots & 1
  \end{pmatrix},
  \quad 
  D = \begin{pmatrix}
    \sqrt{P_{0}} & 0 & \cdots & 0  \\
    0 & \sqrt{P_{1}} & \cdots & 0  \\
    \vdots & \vdots & \ddots & \vdots  \\
    0 & 0 & \cdots & \sqrt{P_{k-1}} 
  \end{pmatrix},
  \quad 
  V = \left( \begin{array}{c} \alpha^{k-1} \\ \alpha^{k-2} \\
      \vdots \\ 1 \end{array} \right).
\end{equation*}

Note that $\phi < 1$ for any $k \geq 1$ because $\phi = 1 - v_{k-1}$.
We establish first the following lemma:
\begin{lem}
  \label{lem:one}
  Let $U$ be a $k \times k$ diagonal matrix with non-negative entries
  $\{x_i\}_{i=0}^{k-1}$ on the diagonal. Among all the permutations of
  the $\{x_i\}$'s, the one that maximizes $V^T \left[A + U
  \right]^{-1} V$ is one where the diagonal entries are in
  non-increasing order.
\end{lem}

\begin{proof}
  Assume that $\{x_i\}_{i=0}^{k-1}$ are in the following order: $0 \le
  x_0 \le x_1 \le \cdots \le x_{k-1}$. We prove in what follows that
  $U = diag(x_{k-1}, x_{k-2}, \cdots, x_0)$ maximizes $V^T \left[A + U
  \right]^{-1} V$ using induction on $k$. To highlight the dependence
  on $k$ we denote $\varphi_k = V_k^T \left[A_k + U_k \right]^{-1}
  V_k$, which is a positive quantity due to the positive definiteness
  of $\left[A_k + U_k \right]^{-1}$.

  \paragraph{Base Cases}
  For $k = 1$, $\varphi_1 = \frac{1}{1+x_0}$ and the statement holds.
  Examine now the case $k = 2$:
  \begin{align*}
    U = diag(x_0, x_1) & \Longrightarrow \varphi_2 = \frac{(\alpha^2 x_1 + x_0
      - \alpha^2 + 1)}{(x_0+1)(x_1+1)-\alpha^2} \\
    U = diag(x_1, x_0) & \Longrightarrow \varphi_2 = \frac{(\alpha^2 x_0 + x_1 - \alpha^2 +
      1)}{(x_1+1)(x_0+1)-\alpha^2}.
  \end{align*}
  Since $\alpha < 1$ and $x_1 \geq x_0$, the second value is larger.

  \paragraph{Induction Step}
  Suppose the property holds true up to $k-1$ ($k-1 \geq 2$) and we
  prove in what follows that it holds true for $k$:
  \begin{equation*}
    \varphi_{k} = \begin{pmatrix}
      \alpha^{k-1} & \alpha^{k-2} & \cdots & 1
    \end{pmatrix}
    \;\left[A_{k} + U_{k}\right]^{-1} \;
    \begin{pmatrix}
      \alpha^{k-1} \\
      \alpha^{k-2} \\
      \vdots \\
      1
    \end{pmatrix},
  \end{equation*}
  where $A_{k}$ and $U_{k}$ are square matrices of size $k$. We prove
  that $\varphi_{k}$ is maximized when $\{x_i\}_{i=0}^{k-1}$ are
  placed in non-increasing order on the diagonal matrix $U_{k}$.  The
  proof proceeds as follows: We first ``fix'' $x_{k-1}$ on the last
  diagonal entry of $U_{k}$ and prove that $\{x_i\}_{i=0}^{k-2}$
  should be in a non-increasing order to maximize $\varphi_{k}$. Next,
  we ``fix'' $\{x_i\}_{i=0}^{k-3}$ on the first $(k-2)$ diagonal
  entries of $U_{k}$ and we prove that, if $x_{k-2} \le x_{k-1}$,
  having $U = diag\{x_0, x_1, \cdots, x_{k-3}, x_{k-1}, x_{k-2}\}$
  (versus $U = diag\{x_0, x_1, \cdots, x_{k-3}, x_{k-2}, x_{k-1}\}$)
  gives us a larger value of $\varphi_{k}$, completing the proof.
  
  $\bullet$ Using a block form, we write $[A_{k} + U_{k}]$ as:
  \begin{equation*}
    A_{k} + U_{k} = \begin{pmatrix}
      E & F \\
      F^{T} & G
    \end{pmatrix},
  \end{equation*}
  where 
  \begin{equation*}
    E = \begin{pmatrix}
      x_0+1 & \alpha & \cdots & \alpha^{k-2} \\
      \alpha & x_1+1 & \cdots & \alpha^{k-3} \\
      \vdots & \vdots & \ddots & \vdots \\
      \alpha^{k-2} & \alpha^{k-3} & \cdots & x_{k-2}+1
    \end{pmatrix}, \quad
    F = \begin{pmatrix}
      \alpha^{k-1} \\
      \alpha^{k-2} \\
      \vdots \\
      \alpha
    \end{pmatrix} = \alpha V_{k-1}
    \quad \& \quad G = \begin{pmatrix}
      x_{k-1} + 1
    \end{pmatrix}.
  \end{equation*}
  
  This allows us to express $[A_{k}+U_{k}]^{-1}$ as~\cite{matrix}:
  \begin{equation*}
    [A_{k}+U_{k}]^{-1} = \begin{pmatrix}
      E^{-1} + E^{-1} F [G - F^{T} E^{-1} F]^{-1} F^{T} E^{-1} &
      -E^{-1} F [G - F^{T} E^{-1} F]^{-1} \\
      -[G - F^{T} E^{-1} F]^{-1} F^{T} E^{-1} &
      [G - F^{T} E^{-1} F]^{-1}
    \end{pmatrix}
  \end{equation*}
  and
  \begin{align*}
    \varphi_{k} = \, & F^T \left[E^{-1} + E^{-1} F [G - F^{T} E^{-1} F]^{-1} F^{T} E^{-1} \right]
    F 
    -[G - F^{T} E^{-1} F]^{-1} F^{T} E^{-1}
    F \\
    & - F^T E^{-1} F [G - F^{T} E^{-1} F]^{-1} + [G - F^{T} E^{-1} F]^{-1},
  \end{align*}
  which reduces to:
  \begin{align*}
    \varphi_{k} & = F^{T} E^{-1} F + \frac{\left( 1 - F^{T} E^{-1} F \right)^2}{\left[ x_{k-1} + 1 - F^{T} E^{-1} F \right]}
    = F^{T} E^{-1} F + \frac{\left(x_{k-1} + 1 - F^{T} E^{-1} F  - x_{k-1} \right)^2}{\left[ x_{k-1} + 1 - F^{T} E^{-1} F \right]} \\
    & = \left(- x_{k-1} + 1 \right) + \frac{x_{k-1}^2 }{\left[ x_{k-1} + 1 - F^{T} E^{-1} F \right]}.
  \end{align*}

  The scalar $F^{T}E^{-1}F$ is equal to $\alpha^2
  \varphi_{k-1}$. Indeed, $F$ is of size $(k-1) \times 1$ and equal to
  $\alpha V_{k-1}$, and $E$ is a $(k-1) \times (k-1)$ sub-matrix of
  the form $[A_{k-1}+U_{k-1}]$. Since $\alpha < 1$, the scalar
  $F^{T}E^{-1}F$ is less than one
  and the denominator is a positive quantity. Therefore, with
  $x_{k-1}$ fixed, $\varphi_{k}$ is maximized when $F^{T}E^{-1}F$ is
  maximized. By the induction step, with a fixed $x_{k-1}$ the
  remaining $x_i$'s should be ``placed'' in decreasing order on the
  diagonal of $E$ --and $U$-- to maximize $\varphi_{k}$.
  
  $\bullet$ Now fix $\{x_i\}_{i=0}^{k-3}$. We prove that with $x_{k-2}
  \le x_{k-1}$, $U = diag\{x_0, x_1, \cdots, x_{k-3}, x_{k-1},
  x_{k-2}\}$ gives us a larger value for $\varphi_{k}$.  To do this,
  we consider a different decomposition of the matrix $[A_{k} +
  U_{k}]$,
  \begin{equation*}
    A_{k} + U_{k} = \begin{pmatrix}
      E & F \\
      F^{T} & G
    \end{pmatrix}
  \end{equation*}
  where now
  \begin{equation*}
    E = \begin{pmatrix}
      x_0+1 & \alpha & \cdots & \alpha^{k-3} \\
      \alpha & x_1+1 & \cdots & \alpha^{k-4} \\
      \vdots & \vdots & \ddots & \vdots \\
      \alpha^{k-3} & \alpha^{k-4} & \cdots & x_{k-3}+1
    \end{pmatrix}
    \, F = \begin{pmatrix}
      \alpha^{k-2} & \alpha^{k-1} \\
      \alpha^{k-3} & \alpha^{k-2} \\
      \vdots & \vdots \\
      \alpha & \alpha^2
    \end{pmatrix}
    \, \& \, 
    G = \begin{pmatrix}
      x_{k-2} + 1 & \alpha \\
      \alpha & x_{k-1} + 1
    \end{pmatrix}.
  \end{equation*}

  Since $F = \left( \alpha V_{k-2} \quad \alpha^2 V_{k-2} \right)$,
  \begin{align*}
    \varphi_{k} = \, & \alpha^4 V_{k-2}^T \biggl[ E^{-1} + E^{-1} F \left[G - F^{T} E^{-1} F \right]^{-1} 
    F^{T} E^{-1} \biggr] V_{k-2} \nonumber \\
    & - 2 \alpha^2 \begin{pmatrix}
      \alpha & 1
    \end{pmatrix}
    \left[ G - F^{T} E^{-1} F \right]^{-1} F^{T} E^{-1} V_{k-2} 
    + \begin{pmatrix}
      \alpha & 1
    \end{pmatrix}
    \left[ G - F^{T} E^{-1} F \right]^{-1} 
    \begin{pmatrix}
      \alpha \\
      1
    \end{pmatrix}
  \end{align*}

  Noting that $\varphi_{k-2} = V_{k-2}^T E^{-1} V_{k-2}$,
  \begin{align*}
    \varphi_{k} = \, & \alpha^4 \varphi_{k-2} + \alpha^6 \varphi_{k-2}^2
    \begin{pmatrix}
      1 & \alpha
    \end{pmatrix}
    \left[ G - F^{T} E^{-1} F \right]^{-1}
    \begin{pmatrix}
      1 \\ \alpha
    \end{pmatrix} \\
    & - 2 \alpha^3 \varphi_{k-2}
    \begin{pmatrix}
      \alpha & 1
    \end{pmatrix}
    \left[ G - F^{T} E^{-1} F \right]^{-1}
    \begin{pmatrix}
      1 \\ \alpha
    \end{pmatrix}
    + \begin{pmatrix}
      \alpha & 1
    \end{pmatrix}
    \left[ G - F^{T} E^{-1} F \right]^{-1}
    \begin{pmatrix}
      \alpha \\ 1
    \end{pmatrix} \\
    = \, & \alpha^4 \varphi_{k-2} + 
    \underbrace{ \begin{pmatrix}
        \alpha(1-\alpha^2\varphi_{k-2}) & 1-\alpha^4\varphi_{k-2}
      \end{pmatrix} \left[ G - F^{T} E^{-1} F \right]^{-1}
      \begin{pmatrix}
        \alpha(1-\alpha^2\varphi_{k-2}) \\
        1-\alpha^4\varphi_{k-2}
      \end{pmatrix} }_{\xi \, = \, \xi(x_{k-2}, x_{k-1})}
  \end{align*}
  
  Examining $\xi$,
  \begin{align}
    \xi & = \begin{pmatrix}
      \alpha-\alpha^3\varphi_{k-2} & 1-\alpha^4\varphi_{k-2}
    \end{pmatrix} \hspace{-0.1cm}
    \begin{pmatrix}
      x_{k-2}+1-\alpha^2\varphi_{k-2} & \alpha- \alpha^3\varphi_{k-2} \\
      \alpha - \alpha^3\varphi_{k-2} & x_{k-1}+1-\alpha^4\varphi_{k-2}
    \end{pmatrix}^{-1} \hspace{-0.1cm}
    \begin{pmatrix}
      \alpha-\alpha^3\varphi_{k-2} \\
      1-\alpha^4\varphi_{k-2}
    \end{pmatrix} \nonumber \\
    & = \frac{(\alpha-\alpha^3\varphi_{k-2})^2 x_{k-1} + (1-\alpha^4\varphi_{k-2})^2 
      x_{k-2} + (1-\alpha^2\varphi_{k-2})(1-\alpha^4\varphi_{k-2}) (1 - \alpha^2) }
    {x_{k-2}x_{k-1} + (1-\alpha^4\varphi_{k-2})x_{k-2} + (1-\alpha^2\varphi_{k-2})x_{k-1} + (1-\alpha^2\varphi_{k-2})(1- \alpha^2) }. 
    \label{eq:xid}
  \end{align}

  Checking the two possibilities, $\xi(x_{k-2}, x_{k-1}) - \xi(x_{k-1},
  x_{k-2})$ has the same sign as
  \begin{multline*}
    \left[(\alpha-\alpha^3\varphi_{k-2})^2 x_{k-1} + (1-\alpha^4\varphi_{k-2})^2 x_{k-2} 
      + (1-\alpha^2\varphi_{k-2})(1-\alpha^4\varphi_{k-2}) (1 - \alpha^2) \right] \\
    \left[ x_{k-2}x_{k-1} + (1-\alpha^4\varphi_{k-2})x_{k-1} + (1-\alpha^2\varphi_{k-2})x_{k-2} + (1-\alpha^2\varphi_{k-2})(1- \alpha^2)\right] \\
    - \left[ (\alpha-\alpha^3\varphi_{k-2})^2 x_{k-2} + (1-\alpha^4\varphi_{k-2})^2 x_{k-1} 
      + (1-\alpha^2\varphi_{k-2})(1-\alpha^4\varphi_{k-2}) (1 - \alpha^2) \right] \\
    \left[ x_{k-2}x_{k-1} + (1-\alpha^4\varphi_{k-2})x_{k-2} + (1-\alpha^2\varphi_{k-2})x_{k-1} + (1-\alpha^2\varphi_{k-2})(1- \alpha^2)\right],
  \end{multline*}
  which is zero if $\alpha = 1$ or $x_{k-2} = x_{k-1}$. Assuming $x_{k-2} <
  x_{k-1}$, it is of the same sign as
  \begin{multline*}
    - ( 1 - \alpha^6\varphi_{k-2}^2 ) x_{k-2}x_{k-1}
    - (1-\alpha^4\varphi_{k-2}) (1-\alpha^2\varphi_{k-2}) (x_{k-1} + x_{k-2})
    - (1-\alpha^2\varphi_{k-2})^2(1 - \alpha^2),
  \end{multline*}
  which is negative and hence $\xi(x_{k-2}, x_{k-1}) < \xi(x_{k-1},
  x_{k-2})$. We conclude that, when fixing $\{x_i\}_{i=0}^{k-3}$,
  $\varphi_{k}$ is maximized when the last two diagonal elements
  $x_{k-2}$ and $x_{k-1}$ are placed in non-increasing order.

  The final step in the proof is to note that if the diagonal entries
  are not in non-increasing order, then either the first $(k-1)$
  entries are not or the last two entries are not. This contradicts
  the previous two properties.
\end{proof}

Before we state and prove the theorem, a couple of quantities that
will come in handy hereafter are the partial derivatives of
$\xi(x_{k-2},x_{k-1})$ defined in~(\ref{eq:xid}):
\begin{align}
  \frac{\partial }{\partial x_{k-2}} \, \xi \quad & \propto \quad - \alpha^2 (1-\alpha^2\varphi_{k-2})^2 x_{k-1}^2 \label{eq:der1}\\
  \frac{\partial }{\partial x_{k-1}} \, \xi \quad & \propto \quad - \left[ (1-\alpha^2\varphi_{k-2}) (1- \alpha^2) \label{eq:der2}
    + (1-\alpha^4\varphi_{k-2}) x_{k-2} \right]^2,
\end{align}
where the expressions above are those of the respective numerators. We
note that both quantities are non-positive and everything else being
constant, the value of $\xi$ decreases as $x_{k-2}$ or $x_{k-1}$
increases.

\begin{theorem*}
  When maximizing the scalar $\phi$ over all the choices of
  $\{P_j\}_{0}^{k-1}$ such that $\displaystyle \sum_{j=0}^{k-1} P_j =
  k P_{tr}$, the maximum is achieved when all the available power is
  allocated to the last pilot, i.e., $P_j = 0$, for all $0 \le j \le
  (k-2)$ and $P_{k-1} = k P_{tr}$.
\end{theorem*}

\begin{proof}
  We start by imposing a lower bound on the powers $\{P_j\}$'s. More
  precisely, for some small enough $\epsilon > 0$, we assume that $P_j
  = \epsilon + P'_j$ and
  \begin{equation*}
    D = \begin{pmatrix}
      \sqrt{\epsilon + P'_{0}} & 0 & \cdots & 0 \\
      0 & \sqrt{\epsilon + P'_{1}} & \cdots & 0 \\
      \vdots & \vdots & \ddots & \vdots \\
      0 & 0 & \cdots & \sqrt{\epsilon + P'_{k-1}}
    \end{pmatrix},
  \end{equation*}
  and we optimize over the $\{P'_j\}$'s subject to the constraint
  \begin{equation}
    \label{eqn:ubound}
    \sum_{j=0}^{k-1} P'_j \le k P_{tr} - k \epsilon.
  \end{equation}
  
  The diagonal matrix $D$ is non-singular, allowing us to express the
  objective function $\phi$ as:
  \begin{align*}
    \phi = V^T \left[ A + \sigma_N^{2} D^{-1} D^{-1} \right]^{-1} V.
  \end{align*}
  
  Applying the result of Lemma~\ref{lem:one} with $U = \sigma_N^{2}
  D^{-1}D^{-1}$ --and diagonal entries $x_i = \sigma_N^{2}
  \frac{1}{\epsilon + P'_i}$, yields that the optimal $\{ P'_j\}$'s
  have to be non-decreasing. Additionally, the
  derivative~(\ref{eq:der2}) indicates
  that the upperbound~(\ref{eqn:ubound}) will be tight. Indeed, fixing
  $\{P'_0, \cdots, P'_{k-2}\}$ (or equivalently $\{x_0, \cdots,
  x_{k-2}\}$) and increasing $P'_{k-1}$ (or equivalently decreasing
  $x_{k-1}$) will increase $\phi (=\varphi_{k})$. This asserts that
  the power on the last pilot should be as large as possible so that
  the upper bound~(\ref{eqn:ubound}) is met with equality.

  The derivatives~(\ref{eq:der1}) and~(\ref{eq:der2}) allow us to make
  an even stronger statement: If $\{P'_0, \cdots, P'_{k-3}\}$ are
  fixed, among the choices of $P'_{k-2}$ and $P'_{k-1}$ such that
  \[
  P'_{k-2} + P'_{k-1} \leq  k P_{tr} - k \epsilon - \sum_{j=0}^{k-3} P'_j \eqdef M,
  \]
  the one that maximizes $\varphi_{k}$ is $P'_{k-2} = 0$ and $P'_{k-1} =
  M$.

  %
  
  Indeed, since the bound will be met with equality and $P'_{k-2}$ is
  less or equal to $P'_{k-1}$ (by the result of Lemma~\ref{lem:one}),
  we let $P'_{k-2} = p$ and $P'_{k-1} = M - p$ and optimize over $p
  \in [0,M/2]$. Equivalently, $x_{k-2} = \sigma_N^{2}
  \frac{1}{\epsilon + p}$, $x_{k-1} = \sigma_N^{2} \frac{1}{\epsilon +
    M - p}$, and since $\varphi_{k-2}$ is fixed, the derivative of
  $\varphi_{k}$ with respect to $p$ is
  \begin{align*}
    \frac{d}{dp} \varphi_{k} & = \frac{\partial \xi}{\partial x_{k-2}} \cdot \frac{d x_{k-2}}{dp} + \frac{\partial \xi}{\partial x_{k-1}} \cdot \frac{d x_{k-1}}{dp} \\
    & \propto \frac{\alpha^2 (1-\alpha^2\varphi_{k-2})^2} {(\epsilon + p)^2 (\epsilon + M - p)^2} - \frac{\left[ (1-\alpha^2\varphi_{k-2}) (1- \alpha^2) 
        (\epsilon + p) / \sigma_N^{2} + (1-\alpha^4\varphi_{k-2}) \right]^2}{(\epsilon + p)^2 (\epsilon + M - p)^2 },
  \end{align*}
  which is of the same sign as 
  \begin{align*}
    & \alpha (1-\alpha^2\varphi_{k-2}) - \left[ (1-\alpha^2\varphi_{k-2}) (1- \alpha^2) 
      (\epsilon + p) / \sigma_N^{2} + (1-\alpha^4\varphi_{k-2}) \right] \\
    \leq & \, \alpha (1-\alpha^2\varphi_{k-2}) - (1-\alpha^4\varphi_{k-2})
    =  \, (-1 - \alpha^3 \varphi_{k-2}) ( 1 - \alpha) \leq 0,
  \end{align*}
  for any $p$ and therefore the maximum is attained when $p = 0$.
  Said differently, $\varphi_{k}$ is maximum when $P'_{k-2} = 0$ and
  $P'_{k-1} = M$.

  
  By Lemma~\ref{lem:one}, the optimal values of $P'_i$ for all $i \in
  \{0, 1, \cdots, k-3\}$ are less or equal to $P'_{k-2}$. Since for an
  optimal power allocation $P'_{k-2}$ is zero, $P'_i = 0$ for all $i
  \in \{0, 1, \cdots, k-2\}$ and $P'_{k-1} = k P_{tr} - k \epsilon$.

  Finally, the same previous observations show that the smaller the
  $\epsilon$ the larger $\phi$ is. Consequently, taking the limit as
  $\epsilon$ goes to zero yields the optimal solution and the proof of
  the theorem is complete.
\end{proof}

\bibliographystyle{IEEEtran}
\bibliography{biblio}

\begin{thebibliography}{10}
\providecommand{\url}[1]{#1}
\csname url@samestyle\endcsname
\providecommand{\newblock}{\relax}
\providecommand{\bibinfo}[2]{#2}
\providecommand{\BIBentrySTDinterwordspacing}{\spaceskip=0pt\relax}
\providecommand{\BIBentryALTinterwordstretchfactor}{4}
\providecommand{\BIBentryALTinterwordspacing}{\spaceskip=\fontdimen2\font plus
\BIBentryALTinterwordstretchfactor\fontdimen3\font minus
  \fontdimen4\font\relax}
\providecommand{\BIBforeignlanguage}[2]{{%
\expandafter\ifx\csname l@#1\endcsname\relax
\typeout{** WARNING: IEEEtran.bst: No hyphenation pattern has been}%
\typeout{** loaded for the language `#1'. Using the pattern for}%
\typeout{** the default language instead.}%
\else
\language=\csname l@#1\endcsname
\fi
#2}}
\providecommand{\BIBdecl}{\relax}
\BIBdecl

\bibitem{chester1}
A.~Lapidoth and S.~Shamai, ``Fading channels: how perfect need \"perfect side
  information\" be?'' \emph{IEEE Trans. on Inf. Theory}, vol.~48, no.~5, pp.
  1118--1134, May 2002.

\bibitem{chester2}
A.~Vakili, M.~Sharif, and B.~Hassibi, ``The effect of channel estimation error
  on the throughput of broadcast channels,'' in \emph{IEEE Intl. Conf.
  Acoustics, Speech \& Sig. Processing, ICASSP}, vol.~4, May 2006, pp. IV--IV.

\bibitem{chester3}
J.~Wang, M.~Li, Y.~Zhang, and Q.~Zhou, ``Effect of channel estimation error on
  the mutual information of mimo fading channels,'' in \emph{Intl. Conf.
  Wireless Comm., Networking \& Mobile Comp., WiCOM}, Oct 2008, pp. 1--4.

\bibitem{chester4}
A.~Lozano, R.~Heath, and J.~Andrews, ``Fundamental limits of cooperation,''
  \emph{IEEE Trans. on Inf. Theory}, vol.~59, no.~9, pp. 5213--5226, Sept 2013.

\bibitem{Cav91}
J.~K. Cavers, ``An {A}nalysis of {P}ilot {S}ymbol {A}ssisted {M}odulation for
  {R}ayleigh {F}ading {C}hannels,'' \emph{IEEE Trans. on Veh. Technol.},
  vol.~40, no.~11, pp. 686--693, Nov. 1991.

\bibitem{Gold97}
A.~J. Goldsmith and S.~G. Chua, ``Variable-rate variable-power {MQAM} for
  fading channels,'' \emph{IEEE Trans. on Comm.}, vol.~45, no.~10, pp.
  1218--1230, Oct. 1997.

\bibitem{Cai05}
X.~Cai and G.~B. Giannakis, ``Adaptive {PSAM} {A}ccounting for {C}hannel
  {E}stimation and {P}rediction {E}rrors,'' \emph{IEEE Trans. on Wireless
  Comm.}, vol.~4, no.~1, pp. 246--256, Jan. 2005.

\bibitem{ref1}
T.~A. Lamahewa, P.~Sadeghi, R.~Kennedy, and P.~Rapajic, ``Model-based pilot and
  data power adaptation in psam with periodic delayed feedback,'' \emph{IEEE
  Trans. on Wireless Comm.}, vol.~8, no.~5, pp. 2247--2252, May 2009.

\bibitem{ref3}
M.~Agarwal, M.~Honig, and B.~Ata, ``Adaptive training for correlated fading
  channels with feedback,'' \emph{IEEE Trans. on Inf. Theory}, vol.~58, no.~8,
  pp. 5398--5417, Aug 2012.

\bibitem{ref18}
X.~Cai and G.~Giannakis, ``Adaptive psam accounting for channel estimation and
  prediction errors,'' \emph{IEEE Trans. on Wireless Comm.}, vol.~4, no.~1, pp.
  246--256, Jan 2005.

\bibitem{ref16}
A.~Ekpenyong and Y.-F. Huang, ``Feedback constraints for adaptive
  transmission,'' \emph{IEEE Sig. Processing Mag.}, vol.~24, no.~3, pp. 69--78,
  May 2007.

\bibitem{ref4}
S.~Misra, A.~Swami, and L.~Tong, ``Cutoff rate optimal binary inputs with
  imperfect csi,'' \emph{IEEE Trans. on Wireless Comm.}, vol.~5, no.~10, pp.
  2903--2913, Oct 2006.

\bibitem{ref6}
S.~Akin and M.~Gursoy, ``Achievable rates and training optimization for fading
  relay channels with memory,'' in \emph{Conf. Inf. Sc. \& Sys, CISS}, March
  2008, pp. 185--190.

\bibitem{ref9}
K.~Almustafa, S.~Primak, T.~Willink, and K.~Baddour, ``On achievable data rates
  and optimal power allocation in fading channels with imperfect csi,'' in
  \emph{Intl. Symp. Wireless Comm. Sys., ISWCS}, Oct 2007, pp. 282--286.

\bibitem{ref10}
S.~Akin and M.~Gursoy, ``Training optimization for gauss-markov rayleigh fading
  channels,'' in \emph{IEEE Intl. Conf. on Comm., ICC}, June 2007, pp.
  5999--6004.

\bibitem{ref11}
S.~Savazzi and U.~Spagnolini, ``Optimizing training lengths and training
  intervals in time-varying fading channels,'' \emph{IEEE Trans. on Sig.
  Processing}, vol.~57, no.~3, pp. 1098--1112, March 2009.

\bibitem{ref13}
H.~Zhang, S.~Wei, G.~Ananthaswamy, and D.~Goeckel, ``Adaptive signaling based
  on statistical characterizations of outdated feedback in wireless
  communications,'' \emph{Proc. of the IEEE}, vol.~95, no.~12, pp. 2337--2353,
  Dec 2007.

\bibitem{ref15}
D.~Duong, B.~Holter, and G.~Oien, ``Optimal pilot spacing and power in
  rate-adaptive mimo diversity systems with imperfect transmitter csi,'' in
  \emph{IEEE Workshop on Sig. Processing Adv. in Wireless Comm.}, June 2005,
  pp. 47--51.

\bibitem{ref19}
A.~Maaref and S.~Aissa, ``Optimized rate-adaptive psam for mimo mrc systems
  with transmit and receive csi imperfections,'' \emph{IEEE on Trans. Comm.},
  vol.~57, no.~3, pp. 821--830, March 2009.

\bibitem{AMM05}
I.~Abou-Faycal, M.~M\'edard, and U.~Madhow, ``Binary {A}daptive {C}oded {P}ilot
  {S}ymbol {A}ssisted {M}odulation over {R}ayleigh {F}ading {C}hannels without
  {F}eedback,'' \emph{IEEE Trans. on Comm.}, vol.~53, no.~6, pp. 1036--1046,
  June 2005.

\bibitem{ZAF09}
K.~Zeineddine and I.~Abou-Faycal, ``How {M}uch {T}raining is {O}ptimal in
  {A}daptive {PSAM} over {M}arkov {R}ayleigh {F}ading {C}hannels,'' \emph{IEEE
  Intl. Symp. Sig. Processing \& Inf. Technology, ISSPIT}, pp. 366 --371, Dec.
  2009.

\bibitem{Med00}
M.~M\'edard, ``The {E}ffect upon {C}hannel {C}apacity in {W}ireless
  {C}ommunications of {P}erfect and {I}mperfect {K}nowledge of the {C}hannel,''
  \emph{IEEE Trans. on Inf. Theory}, vol.~46, no.~3, pp. 935--946, March 2000.

\bibitem{JakesBook}
W.~C. Jakes, \emph{Microwave {M}obile {C}ommunications}.\hskip 1em plus 0.5em
  minus 0.4em\relax New York: Wiley, 1974.

\bibitem{ATS01}
I.~Abou-Faycal, M.~D. Trott, and S.~Shamai, ``The {C}apacity of
  {D}iscrete-{T}ime {M}emoryless {R}ayleigh-{F}ading {C}hannels,'' \emph{IEEE
  Trans. on Inf. Theory}, vol.~47, no.~4, pp. 1290--1301, May 2001.

\bibitem{Gall87}
R.~Gallager, ``{P}ower {L}imited {C}hannels: {C}oding, {M}ultiaccess, and
  {S}pread {S}pectrum,'' \emph{MIT LIDS}, Nov. 1987.

\bibitem{Verdu90}
S.~Verdu, ``On {C}hannel {C}apacity per {u}nit {c}ost,'' \emph{IEEE Trans. on
  Inf. Theory}, vol.~36, no.~5, pp. 1019--1030, Sep. 1990.

\bibitem{Cheng06}
C.~Luo, ``Communication for {W}ideband {F}ading {C}hannels: on {T}heory and
  {P}ractice,'' Ph.D. dissertation, Massachusetts Institute of Technology, Feb.
  2006.

\bibitem{Baddour05}
K.~E. Baddour and N.~C. Beaulieu, ``Autoregressive {M}odeling for {F}ading
  {C}hannel {S}imulation,'' \emph{IEEE Trans. on Wireless Comm.}, vol.~4,
  no.~4, pp. 1650--1662, July 2005.

\bibitem{BAM04}
A.~Bdeir, I.~Abou-Faycal, and M.~M\'edard, ``Power {A}llocation {S}chemes for
  {P}ilot {S}ymbol {A}ssisted {M}odulation over {R}ayleigh {F}ading {C}hannels
  with no {F}eedback,'' \emph{IEEE Int. Conf. Comm., ICC}, vol.~2, pp.
  737--741, Jun. 2004.

\bibitem{matrix}
K.~B. Petersen and M.~S. Pedersen, ``The matrix cookbook,'' Nov 2012, version
  20121115.

\end{thebibliography}

\end{document}